\documentclass[sigconf,authorversion,screen]{acmart}

\AtBeginDocument{%
  \providecommand\BibTeX{{%
    \normalfont B\kern-0.5em{\scshape i\kern-0.25em b}\kern-0.8em\TeX}}}

\copyrightyear{2024} 
\acmYear{2024} 
\setcopyright{rightsretained} 
\acmConference[KDD '24]{Proceedings of the 30th ACM SIGKDD Conference on Knowledge Discovery and Data Mining}{August 25--29, 2024}{Barcelona, Spain}
\acmBooktitle{Proceedings of the 30th ACM SIGKDD Conference on Knowledge Discovery and Data Mining (KDD '24), August 25--29, 2024, Barcelona,
Spain}\acmDOI{10.1145/3637528.3671618}
\acmISBN{979-8-4007-0490-1/24/08}

\usepackage{amsmath,amsfonts,bm}

\def\eqref#1{equation~\ref{#1}}

\def\1{\bm{1}}

\DeclareMathAlphabet{\mathsfit}{\encodingdefault}{\sfdefault}{m}{sl}
\SetMathAlphabet{\mathsfit}{bold}{\encodingdefault}{\sfdefault}{bx}{n}

\usepackage{hyperref}
\usepackage{url}

\usepackage{comment}
\usepackage{amsmath}

\usepackage{amssymb}
\usepackage{mathtools}
\usepackage{float}
\usepackage{amsthm}
\usepackage{bbm}
\usepackage{multirow}
\usepackage{thm-restate}
\usepackage{wrapfig}
\usepackage{booktabs}
\usepackage{algorithm}
\usepackage{algorithmic}
\usepackage{enumitem}
\usepackage{xspace}
\usepackage{xcolor}
\usepackage{balance}

\newcommand\norm[1]{\left\lVert#1\right\rVert}
\newcommand{\pgr}{PG-Ret\xspace}

\newtheorem{assumption}{Assumption}
\newtheorem{theorem}{Theorem}
\begin{document}

\title{
Optimizing Novelty of Top-k Recommendations using Large Language Models and Reinforcement Learning
}



\author{Amit Sharma}
\affiliation{%
  \institution{Microsoft Research}
  \city{Bengaluru}
  \country{India}
}
\email{amshar@microsoft.com}

\author{Hua Li}
\affiliation{%
 \institution{Microsoft Bing Ads}
 \city{Redmond}
 \country{USA}}
\email{huli@microsoft.com}

\author{Xue Li}
\affiliation{%
 \institution{Microsoft Bing Ads}
 \city{Mountain View}
 \country{USA}}
\email{xeli@microsoft.com}

\author{Jian Jiao}
\affiliation{%
 \institution{Microsoft Bing Ads}
 \city{Redmond}
 \country{USA}}
\email{jiajia@microsoft.com}

\renewcommand{\shortauthors}{Sharma et. al.}
\newcommand{\amit}[1]{\textcolor{blue}{#1 --amit}}

\begin{abstract}
Given an input query, a recommendation model is trained using user feedback data (e.g., click data) to output a ranked list of items. In real-world systems, besides accuracy, an important consideration for a new model is novelty of its top-k recommendations w.r.t. an existing deployed model. However, novelty of top-k items is a difficult goal to optimize a model for, since it involves a non-differentiable sorting operation on the model's predictions. Moreover, novel items, by definition, do not have any user feedback data. Given the semantic capabilities of large language models, we address these problems using a reinforcement learning (RL) formulation where large language models provide feedback for the novel items. However, given millions of candidate items, the sample complexity of a standard RL algorithm can be prohibitively high. To reduce sample complexity, we reduce the top-k list reward to a set of item-wise rewards and reformulate the state space to consist of $\langle$query, item$\rangle$ tuples such that the action space is reduced to a binary decision; and show that this reformulation results in a significantly lower complexity when the number of items is large.  We evaluate the proposed algorithm on improving novelty for a query-ad recommendation task on a large-scale search engine. Compared to supervised finetuning on recent \textless query, ad\textgreater pairs, the proposed RL-based algorithm leads to significant novelty gains with minimal loss in recall.  We obtain similar results on the ORCAS query-webpage matching dataset and a product recommendation dataset based on Amazon reviews.
\end{abstract}





\maketitle

\section{Introduction}
Given a user's profile or an input query, the recommendation problem is to fetch a ranked list of top-k items based on a task-specific goal. We consider the retrieval layer of a recommendation system, where the input is typically millions of candidate items and output is hundreds of ranked items (e.g., k=$200$). For the retrieval layer, semantic relevance is a common goal and models are trained to rank relevant items higher~\citep{guo2022semantic}. 
A common way to train such recommendation models is to use supervised learning on user feedback data such as clicks, using losses such as  contrastive learning that encourage representations of clicked query-item pairs to be closer to each other than negative (or random) query-item pairs~\citep{gao2021simcse}.

However, in addition to relevance, real-world recommendation systems typically  have additional goals to optimize for the predicted top-k items. For instance, an important goal is that a candidate model predicts top-k items that are \textit{novel} w.r.t. the existing deployed models in the system~\cite{kaminskas2016diversity}. Novelty is a desirable property since it can avoid showing repetitive or redundant items to a user, enhance global coverage over all available items, and help avoid any systematic bias towards previously clicked items in the system~\cite{herlocker2004evaluating}.  However, it is difficult to train a model to optimize novelty of top-k items since evaluating novelty   requires a non-differentiable sorting operation to obtain the top-k list.  Moreover, novel items, by definition, do not have any user feedback data corresponding to the particular query. As a result, getting relevance feedback on novel items  requires conducting an online experiment showing exploratory, novel items to the user, which can be a costly and infeasible procedure for many systems. 

\begin{figure*}[t]
\begin{center}
   \includegraphics[width = 1.0\linewidth]{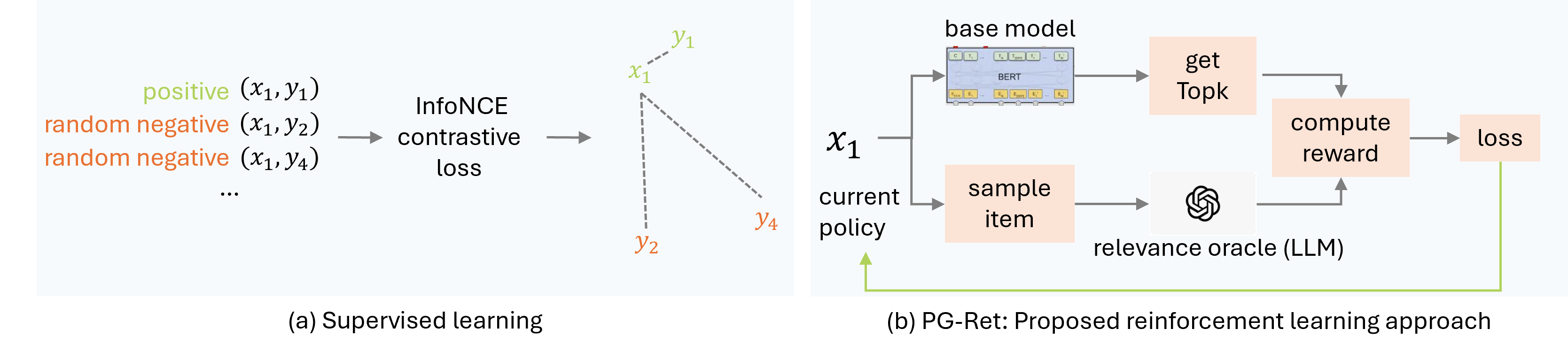}
\end{center}
   \caption{Our goal is to optimize retrieval models for novelty. \textit{Left:} Supervised learning with a contrastive loss is unable to optimize novelty directly since it involves a non-differentiable top-k sorting operation.  \textit{Right:} the proposed RL approach with explicit optimization for the novelty of top-$k$ retrieved items (w.r.t. to a base model) using LLM-based relevance feedback. }
   \label{fig.ps}
\end{figure*}

To obtain scalable feedback for novel items, we utilize recent results that show that large language models (LLMs) like GPT-3.5 or GPT-4 can match or surpass quality of crowd-sourced relevance feedback for recommendation tasks, such as matching user input to ad keywords~\cite{he2023annollm} or matching user history to suggested  movies and games~\cite{hou2023large}. Specifically, we  propose that LLMs can be used as \textit{reward models} to  provide relevance feedback for novel items, using an appropriate task-specific prompt. As a result, (noisy) relevance feedback for novel items can be obtained at scale and be used for both offline model training and evaluation, with human feedback restricted to only evaluating the final model in an A/B test. 

For handling non-differentiability of the novelty metric, we consider a reinforcement learning (RL) formulation, as in past work on optimizing metrics such as diversity~\cite{liu2021diversity} and popularity~\cite{shi2023relieving,stamenkovic2022choosing} of top-k predicted items. Specifically, the retrieval problem can be framed as a one-step RL problem where the input query is the state, the top-k predicted items constitute the action, and their relevance and novelty constitute  the environment reward.  However, 
the large action space in retrieval tasks---equal to the number of available items that can be predicted for an input query (typically millions)---presents a challenge due to the data requirements for RL algorithms.  While we use a policy gradient RL algorithm that  does not explicitly parameterize the action space and hence is better suited for large action spaces,  
the finite sample convergence results of policy gradient algorithms~\citep{mei2020global} show  a direct dependence of policy error on the square of the action space size. Consequently, the resultant algorithms are reported to not work well beyond tens of items~\cite{lu2022promptpg}.

To address the issue of optimizing over large action spaces, we make two contributions. First, rather than the standard practice of training RL algorithms from scratch for recommendation systems~\cite{shi2023relieving,stamenkovic2022choosing,keat2022multiobjective,xie2021hierarchical}, we argue that it is more practical to consider RL as a tool to \textit{finetune} an existing supervised learning model that has been trained on user feedback data. Such an approach has led to substantial gains in training image~\cite{pinto2023tuning} and language~\cite{ouyang2022training} models and we extend it for recommendation systems. We also provide a formal justification for this line of work: assuming that the existing model has a higher probability of  selecting the optimal action (item)  for a query than a uniformly random policy, we show that initializing from an existing model reduces the sample complexity requirements for a policy gradient algorithm. 

Second, we propose a reformulation of the problem that further reduces the sample complexity. Rather than considering a query as the state, we consider the (query, item)  pair as a state. As a result, the action space is converted to a  binary space: whether the item is preferred by the task-specific goal (e.g., relevance and novelty) or not.  Theoretically, our formulation reduces the dependence of convergence error from quadratic to linear in the action space, which is a substantial reduction whenever the action space is large. Crucially, even though the action space is binary given a query-item pair, policy  optimization can use rewards based on the top-$k$ retrieved list from the current policy, which was not possible in supervised learning. At the same time, the reduction leads to an empirical advantage: we can develop a scalable, batch-wise training  algorithm for optimizing the policy using loss functions that are well-suited for large-scale recommendation, such as those based on contrastive learning. We call the resultant algorithm, \textit{\pgr}, Policy Gradient for large-scale Retrieval. 

We evaluate  \pgr on three datasets: a proprietary dataset for query-ad recommendation task, a public dataset for query-webpage matching task, and a  product recommendation dataset based on Amazon reviews. In the first dataset, the task is produce top-k items that are novel compared to a given production model on a commercial search engine. The training data consists of recently collected \textless query, ad\textgreater pairs from the search engine's logs and has 33M actions. We find that \pgr leads to significant improvement (2X-5X) in novelty of top-k retrieved items compared to a supervised model trained on the same training data,  while incurring minimal drop in precision and recall. Online A/B test shows a statistically significant gain of 1\% in retrieval density, the average number of relevant ad keywords fetched per query. The second dataset is a public dataset, ORCAS, for matching queries to webpage titles. We obtain a similar result: \pgr leads to over 2X improvement in novelty for the top-20 items while incurring a small drop in recall and precision.  Finally, \pgr can also be applied to general recommendation scenarios beyond query-based tasks. We apply it to an Amazon user-product recommendation task and find significant gains in novelty compared to supervised finetuning. 

To summarize, we propose a technique based on the policy gradient algorithm that is simple to implement (see Algorithm 1)  and shows the benefit of finetuning existing models with RL for task-specific goals such as novelty. Our contributions include, 
\begin{itemize}
\item A simple, model-agnostic algorithm to finetune retrieval models that works for large action spaces (millions of actions) and arbitrary task-specific goals. 
\item Leveraging relevance feedback from LLMs; thus allowing for optimization of novel items directly since novel items, by definition, do not have user feedback data.
\item Theoretical and empirical justification of the method for increasing novelty in recommendation systems, including evaluation of the method for predicting ad keywords on a commercial search engine. 
\end{itemize}

\section{Related Work}

\subsection{Enhancing novelty of top-k results.}
Outputting novel items compared to existing algorithms is an important goal for recommendation systems~\cite{shani2011evaluating}. A common approach is to consider a  multi-objective optimization problem~\cite{rodriguez2012multiple} with relevance and novelty as the two objectives~\cite{hurley2011novelty,diez2019optimizing}. Recently, reinforcement learning  has been used for novelty optimization since RL can be applied to non-differentiable objectives like novelty~\cite{stamenkovic2022choosing,keat2022multiobjective,xie2021hierarchical}. However, past work uses approximations of novelty, such as an item being less popular~\cite{stamenkovic2022choosing,shi2023relieving,vargas2011rank} or belonging to less popular categories~\cite{xie2021hierarchical},  to enhance novelty,  since truly novel items---those that have not been shown by the current system---would not have any user feedback data for training. To ascertain relevance of truly novel items, heuristics like estimating the curiosity of a user (more curiosity admits more novel items as relevant) have been proposed~\cite{zhao2016much}, but no systematic way to estimate relevance exists.  Consequently, without relevance feedback, there exists a direct tradeoff between optimizing novelty and relevance of a top-k list.  In this work, we use the semantic capabilities of LLMs to propose a  general, scalable method to estimate  relevance of novel items.  As a result, we can directly optimize for the novelty objective without sacrificing relevance.

Note that encouraging novelty in top-k items is different from addressing the  cold-start problem~\cite{sanner2023large,wu2024could}. Novel items are defined wrt. a query whereas cold-start items are defined globally for a system. In many cases, a novel item wrt. a query may be a popular item globally. As a result, methods for the cold-start problem may not directly apply to novelty optimization.

\subsection{RL with large action spaces.}
Given a query, information retrieval can be divided into two stages~\cite{hron2021component}; 1) retrieval of relevant items from a large candidate pool of items; 2) ranking the retrieved items to select a smaller set of items that are shown to the user. RL algorithms find many applications in the ranking phase, including contextual bandits~\citep{li2010contextual}, markov decision processes~\citep{wei2017mdp}, and policy gradient  algorithms~\citep{pan2019policy,xu2020ltrpolicy}. Many of these algorithms introduce a task-specific loss to the ranking problem. However, applying these techniques to the retrieval phase is challenging because of the large number of candidate items (actions).  

Large action spaces is a general challenge for RL beyond retrieval models~\citep{dulac2015deep}, even when using policy gradient, a method that does not explicitly parameterize the action space and thus is better suited for higher number of actions.  For sequential recommendation problems, appropriately parameterizing the policy and using off-policy correction helps in scaling to large actions~\cite{chen2019top}. For one-step RL problems, recent theoretical work in contextual bandits~\citep{zhu2022contextual,lopez2021learning} tries to address the large actions problem. However, the focus of these approaches is on obtaining an optimal solution from scratch, which may be difficult in practice and misses the complementary benefits of supervised learning over user feedback. Instead, we finetune an existing supervised model, as done by \cite{pinto2023tuning} for computer vision models.  

\subsection{LLMs for information retrieval.} Recently, LLMs like GPT-3.5 have been applied to retrieval tasks in a zero-shot setting with encouraging results~\citep{dai2022promptagator,hou2023large,arora2023gar}. Instead of using compute-intensive LLMs directly for the retrieval task, here we aim to use LLMs as reward models to train an efficient, small retrieval model.

\section{Optimizing novelty W/ LLM FEEDBACK}
\subsection{Problem statement: Novelty optimization}
Commonly used retrieval models use a bi-encoder architecture~\citep{gao2021simcse}, where the same neural network model embeds a query and item into a common representation space. The top-k items are selected based on the nearest neighbors to a query, as measured by a suitable distance function over the embeddings (e.g., cosine similarity). The encoder $\phi$ is typically optimized using a variant of contrastive learning, encouraging that positive $<$query,item$>$ pairs in the training set should be closer in embedding space than non-positive pairs. Non-positive pairs may be random pairs~\cite{gao2021simcse} or mined from train data\cite{dahiya2021siamesexml}. Thus, given a set of queries $X$ and items $Z$,  a  train dataset $D \sim \mathcal{D}$ with positive query-item pairs, and a similarity function $\texttt{sim}$, the trained retriever $\hat{\phi}$ can be written as (using the InfoNCE contrastive loss function~\cite{oord2018representation}),
\begin{align} \label{eq:cont-learn}
\hat{\phi} &= \arg \min_\phi -\log \sum_{(x,z)\in \mathcal{D}} \frac{\exp(\texttt{sim}(\phi(x),\phi(z)))}{\sum_{z' \in neg(x)} \exp(\texttt{sim}(\phi(x),\phi(z')))} \nonumber \\
\mathbf{y}(x) &= \texttt{Topk}_{z\in Z} \texttt{sim}(\hat{\phi}(x),\hat{\phi}(z))
\end{align}
where    $\mathbf y = [y_1, y_2, \cdots, y_k]$ refers to the sequence of top-k items returned by a trained encoder $\hat{\phi}$ and $\texttt{Topk}$ is a non-differentiable ranking operation.
At inference time, given a query, the trained encoder $\hat{\phi}$ is used to compute the similarity with each item and the top-k closest items are returned as the prediction.

Given the trained encoder $\hat\phi$, our goal is to finetune the encoder such that the novelty of its top-k items compared to a base model is optimized, while ensuring that the accuracy of the encoder does not decrease substantially. Specifically, we assume that there is an existing base retrieval model $\psi$ deployed in the production recommendation system and we want to deploy another retrieval model that introduces novel items. Both the base model and the new, finetuned model would run in parallel and their results can be merged for the downstream ranking layer. Formally, novelty@k is defined as,
\begin{definition}\label{def:novelty}
\textbf{Novelty:} Given a query $x$, Novelty@k $(x, \phi, \psi, L)$ is defined as the number of items in top-k predictions of a candidate model $\phi$ that do not exist in top-L predictions of the base model $\psi$.  
\end{definition}
Typically, $L$ is set as the number of items that the retrieval layer sends downstream to the ranking layer (e.g., $L=200$ in our experiments). If a candidate model outputs an item that is beyond the top-L of the base model, then it will result in a new item added to the recommendation system's ranking layer. Note that $k$ need not be the same as $L$. For efficiency reasons, we may be interested in novelty at lower $k$ (e.g., $k=50$), under a setting where the candidate model only adds $k$ additional items to be sent to the downstream ranking layer.  
While our definition assumes a single base model, it can be easily extended to multiple deployed models by pooling their top-$L$ predictions for novelty comparison.

Note that novelty of top-k items of a candidate algorithm is different from its diversity~\cite{shani2011evaluating}, which measures the distance of items \textit{within} the top-k list produced by the candidate algorithm.

\subsection{RL formulation using policy gradient}

\texttt{Novelty@k}, as defined above, cannot be optimized directly since it includes a non-differentiable component w.r.t. the encoder's parameters, the top-k sorting operation over items. As in past work on recommendation systems~\cite{stamenkovic2022choosing,keat2022multiobjective,xie2021hierarchical}, a natural solution for non-differentiable rewards is to use reinforcement learning.  Formally, our problem can be formulated as a one-step RL setting. The query is considered as the \textit{state} and the set of top-k items as the \textit{action}. A policy $\pi: X \rightarrow \{Z_k : Z_k \subseteq Z\}$ is a function that outputs a set of top-k items given a query. For each action selected by the policy, the environment provides reward feedback on the $<$ state,action $>$ pair. For instance, in the query-ad recommendation task, user's input query is the state, the top-k predicted ads are the action by the policy, and the environment provides reward based on the query and its top-k predicted ads, e.g., based on the novelty and relevance of the predicted ads. Given a reward function $\mathcal{R}$, the task to learn a policy $\pi_\theta$ (parameterized by $\theta$), that optimizes,
\begin{equation} \label{eq:mot-goal}
    \max_\theta \mathbf{E}_{x\sim \mathcal{D}} \mathbf{E}_{\mathbf{y} \sim \pi_\theta(x)} \mathcal{R}(x,\mathbf{y})
\end{equation}

In our work, we consider a pre-trained encoder $f_\theta: X \cup Z \rightarrow \mathbf{R}^d$ to correspond to the initial policy $\pi_\theta^0$. For simplicity, given an encoder/policy at any stage of training, the top-k items  $\mathbf{y}$ are assumed to be independently sampled from the discrete action probability distribution, 
$\pi_\theta: X \rightarrow Z; \text{\ \ }y_j \sim \pi_\theta(z|x)  \forall j \in \{1,2,3..k\}$; 
where $\pi_\theta(z| x) = \operatorname{softmax}_Z \texttt{sim}(f(x), f(z))$. We use $\pi_\theta (\mathbf{y}|x)$ to denote the top-k items generated using the $\pi_\theta$ policy.

As in past work on using RL in recommendation systems~\cite{pan2019policy,montazeralghaem2020reinforcement,chen2019top}, we use a policy gradient algorithm to optimize the novelty reward. Policy gradient algorithms are well-suited for recommender systems since they  do not explicitly parameterize the action space (thus allowing a large number of actions) and can directly update the parameters of the encoder $f_\theta$. Specifically, we use the REINFORCE~\citep{williams1992simple} algorithm that depends on a Monte Carlo approximation of the reward expectation from Eqn.~\ref{eq:mot-goal}.

 \begin{equation}
       \mathbf{E}_{x\sim \mathcal{D}} \mathbf{E}_{\mathbf{y} \sim \pi_\theta(x)} \mathcal{R}(x,\mathbf{y}) \approx \frac{1}{B} \sum_{i=1}^{B} \mathcal{R}(x^{(i)}, \mathbf{y}^{(i)})
 \end{equation}
 where $y^{(i)}_j \sim \pi_\theta(x^{(i)}) \forall j \in \{1,2,3, \cdots, k\}$ and B is the batch size.  The loss gradient is given by, 
\begin{align} \label{eq:pgr-grad}
\nabla & \mathbf{E}_{x\sim \mathcal{D}} \mathbf{E}_{\mathbf{y} \sim \pi_\theta(x)} \mathcal{R}(x,\mathbf{y}) = \mathbf{E}_{x\sim \mathcal{D}} \mathbf{E}_{\mathbf{y} \sim \pi_\theta(x)} \mathcal{R}(x,\mathbf{y}) \nabla \log{\pi_\theta(\mathbf{y}|x)} \nonumber \\
& \approx \frac{1}{B} \sum_{i=1}^{B} \mathcal{R}(x^{(i)}, \mathbf{y}^{(i)}) \nabla \log{\pi_\theta(\mathbf{y^{(i)}}|x^{(i)})} 
\end{align} 
 where  $y^{(i)}_j \sim \pi_\theta(x^{(i)}) \forall j \in \{1,2,3, \cdots, k\}$. Since the reward is one-step, the above optimization can be interpreted to have a simple goal: increase the probability of items that occur in a k-sequence with high reward, and decrease the probability of items that occur in a k-sequence that obtains low reward. 
Note that we use the REINFORCE algorithm for simplicity,  but any other policy gradient algorithm can be used.
\subsection{LLMs make novelty optimization practical}
While the formulation is reasonable, there is a key limitation: novel items, by definition, do not have any user feedback data since they were never shown to users by the base production model for that query. Hence, if we use only the log-based training data for Eqn.~\ref{eq:pgr-grad}, relevance feedback for the novel items would be missing and thus no novel item's reward can be computed. 

In this paper, we show that LLMs help avoid this limitation by providing relevance feedback for novel query-item pairs. LLMs such as GPT-3.5 and GPT-4 have been shown to provide a substantial improvement in semantic understanding compared to existing encoder-based models.  For instance, recent work shows that GPT-3.5 can match or surpass crowd-sourced human labellers in their accuracy on labelling~\cite{he2023annollm} or ranking~\cite{hou2023large} retrieval outputs for relevance.  In addition, even though the relevance criterion for a correct recommendation may differ for different tasks, we need not train separate relevance models. For example,  while the criteria for matching books to other books may be different from matching  queries to ads, which in turn may be slightly different for matching queries to webpages;  a single LLM like GPT-3.5 can be used for providing relevance feedback for all these domains by simply changing the prompt. Hence, the accuracy and universality of LLMs make it possible to obtain relevance feedback for arbitrary query-item pairs and optimize novelty directly for retrieval tasks.

\subsection{The challenge with large action spaces}
While LLMs provide a solution to the problem of relevance feedback for novel query-item pairs, another key challenge is the large number of potential query-item pairs to label since available items are typically in the millions. Below we show that the number of relevance reward samples needed for a policy gradient algorithm 
to obtain a desired accuracy (\textit{sample complexity})  increases proportional to the square of the number of actions. The result is based on applying finite sample convergence bounds~\cite{mei2020global} for policy gradient algorithms to the retrieval setup.

Let $\pi^*$ denote the optimal policy that maximizes Eqn~\ref{eq:mot-goal} and $t$ refer to the steps of the optimization. We consider a one-step Markov Decision Process where each query corresponds to a state and the actions are the top-k predicted items. Hence, for each state, we only need to  optimize the current reward. We assume a uniform probability for sampling each state (query).  

\begin{assumption}
 [from \cite{mei2020global}](Sufficient exploration). The initial state $\mu$
distribution satisfies $\min_s \mu(s) > 0$. 
\end{assumption}
\begin{assumption}
    For each state, the initial policy's probability for selecting the optimal action is better than random (uniform probability) within some multiplicative constant $\rho \geq 1$:  $\pi_{\theta_0}(\mathbf{y}^*(x)|x) > \frac{1}{\rho A} \text{ } \forall x \in {D}$ where $\mathbf{y}^*(x) := \arg \max_{\mathbf{y}} \pi^*(\mathbf{y}|x)$. 
\end{assumption}
\begin{restatable}{proposition}{firstprop}\label{thm-orig} 
    Let Assumptions 1 and 2 hold and let $\{ \theta_t\}_{t\geq 1}$ be
generated using the standard policy gradient update: $\theta_{t+1} \leftarrow \theta_t + \eta \frac{\partial V^{\pi_{\theta_t}}(\mu)}{\partial \theta_t}$ where $V^{\pi_{\theta_t}}(\mu) = \mathbf{E}_{x\sim \mathcal{D}}\mathbf{E}_{\mathbf{y}\sim \pi_{\theta_t}(.|x)}\mathcal{R}(x, \mathbf{y})$.   Then, for all $t \geq 1$, with $\eta=1/8$,
\begin{equation}
\mathbf{E}_{x \sim \mathcal{D}}[(\mathbf{\pi}^*(x) - \mathbf{\pi}_{\theta_t}(x) )^T \mathbf{r}(x)] \leq \frac{16S A^2 \rho^2}{t} \norm{\frac{1}{\mu}}_{\infty}
\end{equation}
where  $S=|X|$ is the number of states and $A$ is the number of actions. 
\end{restatable}
Proof is in Appendix. The proof uses Thm. 4 from \cite{mei2020global} for a Markov Decision Process and  adapts it to the single-step problem and additionally uses Assumption 2 to express the bound in terms of $\rho$ and $A$. Since the policy outputs a list of top-k items, the total number of possible actions is $A = C(|Z|, k)$, indicating a high sample complexity. In the next section, we describe how we reduce the effective sample complexity and derive a practical algorithm.

\section{Large Action Policy Gradient}
Proposition~\ref{thm-orig} shows that a naive application of the policy gradient algorithm will be slow to converge to the optimal reward-maximizing solution due to a combinatorially large action space. First, we show that the combinatorial action space $A=C(|Z|,k)$ can be decomposed into item-wise action space, $A=|Z|$ whenever the reward over top-k items can be decomposed as an additive sum over item-wise rewards. We show that most common novelty reward functions satisfy this property. Second, we provide a reformulation of the RL problem that reduces the action space to a binary decision and further increases the rate of convergence. 

\subsection{Reduction from Top-k to item-wise rewards}
As stated in Section 3,  the reward function is a combination of novelty and relevance rewards. We assume that both novelty and relevance rewards compose additively. That is, given a query $x$, 
$Novelty@k (x, \phi, \psi, L) = \sum_{j=1}^k Novelty(\langle x, z_j\rangle, \phi, \psi, L)$, and we have
$Relevance@k(x, \phi) = \sum_{j=1}^k Relevance(\langle x,z_j \rangle,\phi)$; 
where $z_j$ are the individual item predictions and $Relevance@k$ is typically recall@k or precision@k. As a result, we can reformulate the action space to consist of individual items, $A=|Z|$. In expectation, maximizing the novelty and relevance reward function for $j\in[1,k]$ separately would imply maximizing the top-k reward. 

While the action space is reduced in size, a key benefit of our formulation is that the  reward can still be a function of the top-k retrieved items from some model. 
This is because the environment can decide the reward based on whether the item is a part of the top-k items for the query. Recall that, given a state (query) $x$ and an action (item) $z \in Z$, the novelty reward is dependent on whether the $z$ is part of the top-k items returned by the base model $\psi$. As an example, we provide a simple reward function combining item-wise relevance and novelty w.r.t.  $\psi$. Given any query $x$ and item $z$, and a relevance oracle, $\texttt{Rel}$ (i.e., an LLM), the reward is given by,
\begin{equation}\label{eq:reward-eg}
    \mathcal{R}_b(x, z) = \begin{cases}
     1 & \text{ if } \operatorname{Rel}(x,z)=1 \text{ and \ } z \not \in topk(x, \psi) \\
    -1 & \text{ if } \operatorname{Rel}(x,z)=0 \text{ and \ } z \in topk(x, \psi)\\
    -0.5 & \text{ if } \operatorname{Rel}(x,z)=1 \text{ and \ } z \in topk(x, \psi) \\
    -1 & \text{ if } \operatorname{Rel}(x,z)=0 \text{ and \ } z \not \in topk(x, \psi) 
    \end{cases} 
\end{equation}
We can see that the reward function penalizes irrelevant items and relevant items that are not novel w.r.t. the base model, while encouraging relevant items that are novel. However, in practice, the LLM-based relevance function may be noisy. In particular, as higher relevance items are likely to be sampled more, false negatives can be a problem. 
Therefore, for noisy oracles, we propose a simpler reward function that is only activated for relevant items but not for irrelevant items predicted by the oracle. 
\begin{equation}\label{eq:rewardfn}
    \mathcal{R}_b(x, z) = \begin{cases}
     1 & \text{ if } \operatorname{Rel}(x,z)=1 \text{ and \ } z \not \in topk(x, \psi) \\
    -0.5 & \text{ if } \operatorname{Rel}(x,z)=1 \text{ and \ } z \in topk(x, \psi) \\
    0 & \text{ otherwise}
    \end{cases} 
\end{equation}
\subsection{Reduction to binary action space}
To further improve the convergence rate, 
 we consider a different RL formulation where the state is $\langle$query, item$\rangle$ pair and the policy outputs the probability of selecting the item. The item-wise reward function changes slightly to accommodate the two actions of selecting the item (1) or not (0). Assuming the same reward logic from Eqn.~\ref{eq:reward-eg}, the reward function becomes, $\mathcal{R}((x, z), a) =  a\mathcal{R}_b(x, z) + (1-a)(-\mathcal{R}_b(x, z))$. Intuitively, if the policy selects the item for a query, then its reward for the action is proportional to $\mathcal{R}_b$. Otherwise, if it does not select the item, then its reward is proportional to negative of $\mathcal{R}_b$.  
 The corresponding gradient is, 
 {\small
\begin{align} \label{eq:new-model}
& \nabla  \mathbf{E}_{x,z\sim \mathcal{D}}  \mathbf{E}_{\mathbf{a} \sim \pi'_\theta(x,z)} \mathcal{R}((x,z), a)  \\
& \approx \frac{1}{B} \sum_{i=1}^{B} [a^{(i)}\mathcal{R}_b(x^{(i)},z^{(i)}) - (1-a^{(i)})\mathcal{R}_b(x^{(i)},z^{(i)})]\nabla \log{\pi'_\theta(a^{(i)}|x^{(i)}, z^{(i)})} \nonumber
\end{align}} 
where $\pi'_\theta(a=1| x,z) =  \pi_\theta(x,z) = \operatorname{softmax}_Z \texttt{sim}(f_\theta(x), f_\theta(z))$ and $\pi'_\theta(a=0| x,z) = 1-\pi_\theta(x,z) = 1 - \operatorname{softmax}_Z \texttt{sim}(f_\theta(x), f_\theta(z))$. 

 With this formulation, the number of states increases to $SA$ but the number of actions reduces to 2. As we show below, the convergence rate is significantly faster since the error now grows linearly with A rather than quadratic.

\begin{restatable}{proposition}{secondprop}\label{prop:second}
With the new formulation, under the assumptions of Proposition 1, for all $t \geq 1$,
\begin{equation}
\mathbf{E}_{x \sim \mathcal{D}}(\mathbf{\pi}^*(x) - \mathbf{\pi}_{\theta_t}(x) )^T \mathbf{r}(x)  \leq \frac{64SA  \rho^2}{t} \norm{\frac{1}{\mu}}_{\infty}
\end{equation}
\end{restatable}
Note that in practice, $\rho$ may be higher for the binary action formulation since there are only 2 actions. Assuming a ``good enough" supervised policy, conservative value for $\rho$ may be $\sim50$, implying that probability of the optimal action under supervised policy is always $\geq 1/(2\times 50)=0.01$. 
Even under this conservative estimate, as long as the number of actions is of the order of millions, $\rho^2 << A$ and hence the convergence rate in Proposition 2 would be significantly faster. 
In other words, as long as $\rho^2 << A$, the binary-action 
policy gradient algorithm will converge at a faster rate. 


For maximizing a reward based on novelty and relevance such as Eqn.~\ref{eq:reward-eg}, the $\rho$ parameter in Proposition~\ref{prop:second} also implies that using the base policy as the initial policy may not be the best choice. The base policy is expected to be significantly better than a random policy at relevance, so it will assign higher than random probability to relevant items---both novel and not novel items. However, within relevant items, by definition, its probability for novel items will be lower than that of the not novel items. To increase the probability of the optimal action under the initial policy even further (and decrease $\rho$), we can use additional training data to finetune the base policy using supervised learning (e.g., InfoNCE loss from Eq.~\ref{eq:cont-learn}). We call this model the \textit{Supervised Model}. As long as the relevance accuracy does not decrease, randomness in the training procedure should shift the order of top-k items under the supervised model. As a result, the supervised model may retain the property of being better than  a random policy at relevance (just like the base policy) but may also yield higher probability to some novel items under the base policy. Therefore, whenever additional training data is available, such a supervised model is expected to have a lower $\rho$ and we recommend to use it as the initial policy.

\subsection{Proposed algorithm: \pgr}
The above discussion indicates the following conditions for a fast convergence rate with policy gradient, \textbf{1)} the action space should be small;  and \textbf{2)} the initial policy should be as close to the optimal as possible. We addressed the first condition through a binary-action formulation of the problem. For the second condition, we proposed initializing the policy optimization with a policy trained using supervised learning (i.e., using Eq.~\ref{eq:cont-learn}).  

Algorithm~\ref{algo:la-pg} shows the resultant training algorithm.  In each batch, $B$ query-item pairs (states) are sampled and the reward is computed for each state. Using Eqn.~\ref{eq:new-model}, we compute the gradient and update the policy parameters $\theta$ at the end of the batch. 
For each application, one needs to specify the sampling procedure for query-item pairs, the reward function, and the estimation procedure for $\pi'(a|x,z)$. 

\noindent \textbf{Sampling query-item states.} Since items are also a part of the state in our proposed formulation, a key question is how to sample the states $\langle x,z \rangle$. To map to the original policy gradient formulation, we can sample queries $x$ randomly from the train dataset for each batch. For each query, we compute the similarity scores with all items using the encoder $f_\theta$ corresponding to the current policy   and then sample  items proportional to their score.  Note that we are not restricted to only sampling items proportional to their similarity score (since items are  not actions now). Therefore, we also add exploration  by sampling  items from another retrieval model (trained independently). For example, for novelty, a natural choice is to sample from predictions of the base model, restricting the sampling only to items that are ranked beyond top-L. Such items will be novel by definition and thus we only need to check for relevance. More generally, we may use any pretrained encoder suitable for the task. Finally, there is a risk that the policy overfits to the reward from the relevance oracle and ``forgets'' the true user feedback data on which the supervised policy was trained~\cite{ouyang2022training}. Thus, we should also add the user feedback data $(x,z) \in {D}$ during training. To this end, in practice, all three sources are combined stochastically: we sample a query $x$ randomly from the dataset and then sample either $z \sim \pi_{{\theta}_t}(x)$ with probability $\alpha$; $z \sim \pi_{{\theta}explore_t}(x)$ with probability $\beta$; or $z \sim P_\mathcal{D}(x)$ with probability $1-\alpha-\beta$. 

\noindent \textbf{Estimating $\pi'(a|x,z)$.} While $\pi'(a|x,z)$ is defined as a softmax operation, computing the softmax over all items is a compute-intensive procedure. Moreover, training procedures for recommendation systems have benefitted from using contrastive losses. Therefore, we we implement an approximation of the softmax using contrastive losses. A straightforward approximation is to only use the items in the current batch as negatives (random in-batch negatives). However, since we initialize with a well-trained supervised policy, we may find that most of the contrastive losses are zero since the chosen item is already closer to the query than random negatives.  To speed up training, we use negatives that are aimed at optimizing novelty. Specifically, for each query, we use the top-M items returned by the base model as negatives. $M$ is a hyperparameter; too high $M$ may destroy all relevance. Therefore, we define two kinds of losses:\textbf{ 1)} \textit{Aggressive-Novelty:} InfoNCE loss from Eqn.~\ref{eq:cont-learn} with top-M negatives (typically M is small, e.g., $M=5$); \textbf{2)} \textit{Conservative-Novelty:} Triplet loss~\cite{le2020contrastive}, bounded $\in [0,1]$ and  margin=0  (effectively $M=1$). 

\noindent \textbf{Efficiency considerations. }Note that to compute top-k items or sample items for a given query, the entire set of actions have to be encoded by the current encoder $f_\theta$. For computationally efficiency, we fix the item encoder to be the initial encoder and only update the query encoder. This avoids having to recompute the item embeddings for each query in the batch. That is, only the query encoder is updated during policy optimization.

\begin{algorithm}
\caption{Large action PG}
   \label{algo:la-pg}
\begin{algorithmic}[1]
   \STATE {\bfseries Input:} Initial Policy $\pi_{\theta}$, Base Model  $\psi$, Training dataset ${D}$, Relevance Oracle $\texttt{rel}$,  Number of epochs $N$, Batch size $B$, Learning rate $\eta$, $\alpha, \beta \in [0,1]$
    \STATE {\bfseries Output:} Trained  policy $\hat{\pi}_\theta$.
    \FOR{epoch=1,2,3..N}
        \FOR{${D}_{batch} \sim  {D}$}
            \STATE $L = 0$ 
            \STATE $i = 0$ 
            \STATE $(x, z) \sim D$ 
            \WHILE{$i < B$}
                \STATE $a \sim \pi'_\theta(x, z)$ // Sample Action 
                \STATE $r_b = \mathcal{R}_b(x, z, \texttt{topk}(x, \psi), \texttt{rel}(x,z))$ // Reward depends on top-k items and relevance oracle
                \STATE $r = r_ba + (-r_b)(1-a)$  
                \STATE $L = L - r \log \pi'_{\theta}(a|x,z)$ 
                \STATE $i = i + 1$ 
                \STATE $x \sim D$ 
                \STATE $z \sim \pi_\theta(x) (\text{wProb } \alpha), \sim \pi_{\theta explore} (x) (\text{wProb } \beta), \text{ or } \sim  D (\text{wProb }  1-\alpha-\beta)$
            \ENDWHILE 
            \STATE $\theta = \theta - \eta \nabla L$ // Can use any gradient optimizer
        \ENDFOR
    \ENDFOR
\end{algorithmic}
\end{algorithm}

\section{Evaluation}
We evaluate \pgr on  increasing novelty of the top-k items of a retrieval model compared to an existing model's output. 

\subsection{Setup: Datasets, Metrics, and Baselines}
\noindent \textbf{Setup.}
We consider the production setup for a recommendation system, wherein there is an existing retrieval algorithm in production. We call this model the \textit{base} model. Typically, this model is trained on millions of query-item clicked pairs collected from log data. We assume access to a small amount of new training data (e.g., query-item pairs collected from log data after the base model has been deployed) and evaluate different ways to produce a new retrieval model with high novelty compared to the base model.  We consider an independent test set for evaluation. As in prior work~\cite{dahiya2021siamesexml}, the candidate pool of items remains the same across train and test datasets, typically of the order of millions of items. \\

\begin{table}[t]
    \centering
    \begin{tabular}{lccc}
    \toprule
     \textbf{Dataset}    & \textbf{Train inputs} & \textbf{Test inputs} & \textbf{\# Actions} \\
     \midrule
     Query-Keyword & 1.23M & 6.2K & 33.2M \\
     ORCAS & 1M  & 8.5K &  1.07M \\
     \bottomrule
    \end{tabular}
    \caption{Statistics for  the Query-Keyword datasets from a commercial search engine and the public ORCAS dataset.} 
    \label{tab:datasets}
\end{table}

\noindent \textbf{Datasets. } On the query-document matching task, we use a dataset from a commercial search engine and a public dataset  (see Table~\ref{tab:datasets}). We also use a Amazon dataset for product recommendation. 

\begin{itemize}[noitemsep,topsep=0pt,parsep=0pt,partopsep=0pt,leftmargin=*]
\item \textbf{Query-keyword recommendation}. This dataset contains queries and ad keywords from a commercial search engine's logs. Given a query, the goal is to predict ad keywords that share the same intent or a more general intent than the query (i.e., \textit{`phrase match'}). The base model is a 4-layer bi-encoder model that has been finetuned with contrastive loss on more than 20M query-keyword pairs. The training data consists of 1M new query-item pairs, along with a candidate pool of over 33 million items. 

\item \textbf{ORCAS}~\cite{craswell2020orcas}. This public dataset contains clicked webpages from Bing for queries from the TREC deep learning challenge. For our evaluation, we consider the search query as the user input and the webpage title as the item. We filter the dataset to remove click data where either the query or the webpage title are empty. The dataset contains 17.5 million query-webpage pairs. To simulate a production setup, we utilize the majority (16.5M) of the dataset for training a supervised model, that acts as the base model. The base model is initialized with SimCSE~\cite{gao2021simcse} and trained for 5 epochs. The remaining 1M are used as \textit{new} training data for optimizing novelty. We also reserve a separate, randomly sampled test set consisting of 8.5K query-keyword pairs.

\item \textbf{AmazonReviews}~\cite{li2023text}. This dataset contains users' product review histories on Amazon. Based on a user's previous reviews, the goal is to predict the next product that they would review. \citet{li2023text} convert it to text-based problem by represent items as a text sequence using their metadata (e.g., “Title: Philips motor Category: Home Appliances Color: Black”). Users are represented as a text concatenation of each item in their profile. For our experiments, we consider the \textit{Industrial \& Scientific} domain consisting of over 11K user histories and 5K products. \citet{li2023text} sort each  user's history  by time and break it down into a train set, validation set (second-last product), and a test set (most recent product in the history). We use the train set for training the base model, which is initialized with the pre-trained \textit{RecFormer}  model from \cite{li2023text}. To simulate the production scenario where users review additional items over time, we use augmented review history for finetuning novelty models that includes both the  train and validation set products (\textit{Finetuning} dataset). In both cases, the test set remains identical and we ensure that the test set is never used during training (since the finetuning stage does not have a validation set for early stopping; we train models for a fixed number of epochs). 
\end{itemize}

\begin{table*}[ht]
\centering
    \begin{tabular}{l|ccc|cc|c}
        \toprule
        \textbf{Model} & \textbf{Novelty@50} & \textbf{Novelty@100} & \textbf{Novelty@200} & \textbf{Recall@100} & \textbf{Recall@200} & \textbf{Precision@3} \\
        \midrule 
        Base Model & -- &  -- & -- & 45.0 & 54.1 & 63.7 \\
        Supervised Finetuning  & 6.6 & 22.7 & 79.1  & \textbf{49.1} & 53.5 & 65.5 \\
        \pgr (Conservative) & \text{14.5} & {39.6} & {107.8} & \text{46.8} & \textbf{56.4} & \textbf{70.0} \\
        \pgr (Aggressive) & \textbf{32.8} & \textbf{70.2} & \textbf{151.8} & 27.6 & 37.7 & 51.8 \\
        \bottomrule
    \end{tabular}
    \caption{Novelty, recall and precision for \pgr  compared to supervised finetuning. \pgr (Conservative) obtains substantially high novelty with equal or better precision. \pgr (Aggressive) provides even higher novelty but with reduced accuracy. }
    \label{tab-novelty}
\end{table*}

\begin{table*}[ht]
    \centering
    \small
    \begin{tabular}{p{0.24\textwidth}|p{0.33\textwidth}|p{0.33\textwidth}}
            \toprule
           \textbf{Query} &  \textbf{Top-5 Ad Keywords (Base Model)} & \textbf{Top-5 Ad Keywords (\pgr (Conservative))} \\
           \midrule
           tankless water heater electric home depot & 
           home depot tankless water heater \par residential electric tankless water heater \par electrical  tankless water heater \par electric tankless water heater \par water heater electric tankless & 
         \textbf{instant water heater electric} \par \textbf{instant electric water heater} \par \textbf{domestic electrical appliances} \par \textbf{instant water heaters electric} \par \textbf{instant electric water heaters} \\
         \midrule
         
         file management pc & 
         computer file management \par
file management in windows \par
file management \par
file management software \par
file management system &  
computer file management \par
files management software \par 
\textbf{software to organize files} \par
file management software \par
\textbf{files software} \\
        \midrule

        hayan sunrise senior living & 
sunrise senior living home \par
sunrise seniors living \par
sunrise senior living wheaton il \par
sunrise senior care living \par
who owns sunrise senior living &   
\textbf{living facilities for seniors} \par
\textbf{senior living facilities} \par
\textbf{elder living} \par
\textbf{elderly living} \par
\textbf{elderly living facilities} \\
        \midrule 
        
         can i finance my child's college with a home equity loan & 
         home equity loan years \par
dcu home equity loan \par
home equity financing \par
10 year home equity loan \par
home equity loan in texas & 
home equality loans \par
\textbf{educational financing} \par
\textbf{household finance} \par
\textbf{home equity laons} \par
home equity laon \\
    \bottomrule
    \end{tabular}
    \caption{Example predictions from base model and \pgr. Bolded items are novel compared to top-200 from the base model.}
    \label{tab:qual}
\end{table*}

\noindent \textbf{Relevance Feedback.} For all datasets, we use GPT-3.5 as the reward model for providing relevance feedback during training. For the first dataset where the goal is produce keywords with more general intent than the query, we use the prompt,

\begin{quote}\it{
Given the query ''\{Query\}'', is the following query expressing a similar but more general intent, ``\{Keyword\}''? Please answer in a single word: Yes/No.}
\end{quote}

For the ORCAS dataset where the goal is to predict the top search result titles, we use the prompt,
\begin{quote}\it{
Given the query, ``\{Query\}'', is it possible that the document titled, ``\{WebPageTitle\}'', is relevant for the user\'s intent. Please answer in a single word: Yes/No.}
\end{quote}
The system prompt is the same in both cases, \textit{``You are an expert in understanding user intent from search engine queries.''}. 
For the AmazonReviews dataset, we use the system prompt, \textit{``You are an expert in understanding user interests from Amazon.com product browsing data.''}. The user prompt is given by,
\begin{quote}\it{
A user is browsing Amazon.com for the following products:\\ 
\{ListOfProducts\} \\ 
\\
Is this a relevant product that targets the same scientific interest?\\
Recommendation: \{CandidateProduct\} \\

Provide a brief reasoning and then the final answer within ``$\langle$Yes/No$\rangle$'' tag.}
\end{quote}

\noindent \textbf{Metrics.} We evaluate \pgr on the following offline metrics. 
\begin{itemize}[noitemsep,topsep=0pt,parsep=0pt,partopsep=0pt,leftmargin=*]
    \item \textbf{Novelty@k}: Novelty of top-k items compared to the base model, as defined in Definition~\ref{def:novelty}. For the query-keyword dataset, we use L=200 for the base model since each retrieval model sends roughly 200 keywords to the downstream ranking layer. For ORCAS and AmazonReviews dataset, we use L=50 since we observe that the relevance of predicted items decreases significantly beyond 50 items. 
    \item \textbf{Recall@k} is  the number of clicked \textless query, item\textgreater pairs from  test data that are in the top-k predictions. As novelty is optimized, recall over the top-k items should not decrease substantially. 
    \item \textbf{Precision@k}: While recall is an important metric, it is dependent on the available clicks in the test data for each query. It is possible that a model makes relevant predictions  but they are not counted since those \textless query, item\textgreater pairs do not exist in the test data. Hence, we use an advanced LLM, GPT-4 as the relevance evaluator for top-k items.  Note that we use a different LLM for evaluation than the one used in training because 1) A more capable model like GPT-4 can provide more reliable relevance feedback; 2) using a different model ensures fairness of evaluation when compared to other baselines. Moreover, we use standardized prompts that are used in the production system for evaluating retrieval models.  For each task, these prompts  have been validated against human-labelled data and they achieve more than 85\% accuracy, thus serving as a reliable approximation of human feedback.   
    The first task uses a prompt tuned for estimating phrase-match relevance of keyword and the second task uses a prompt tuned for estimating the general relevance of a query and a webpage title.
\end{itemize}
In addition, we evaluate \pgr using an A/B experiment for the first task of matching query to relevant ad keywords. 

\noindent \textbf{Baselines. } For each dataset, we compare \pgr to a supervised model initialized with the base model and trained on the same training data using the InfoNCE loss with random negatives (Eq.~\ref{eq:cont-learn}).  
All models are trained using Adam optimizer with a  learning rate of $10^{-5}$ and batch size of 128. \pgr (Aggressive) uses $M=5$. For the AmazonReviews dataset, we use the same hyperparameters as in \cite{li2023text}.

\begin{table*}[t]
\centering
    \begin{tabular}{l|ccc|c|c}
        \toprule
        \textbf{Model} & \textbf{Novelty@10} & \textbf{Novelty@20} & \textbf{Novelty@50}  & \textbf{Recall@50} & \textbf{Precision@1} \\
        \midrule 
         Base Pretrained Model & -- &  -- & --  & \textbf{81.0} & \textbf{70.0}\\
        Supervised Finetuning (3 epochs)  & 0.007 & 0.04 & 0.65    & {79.7} & 69.9 \\
        Supervised Finetuning (10 epochs) & 0.07 & 0.33 & 2.82 & 78.7 & 68.4\\
        \pgr (Conservative) & {0.17} & {0.47}  & 2.64 &  79.8 & 68.7 \\ 
        \pgr(Aggressive) & \textbf{0.33} & \textbf{0.85} & \textbf{4.03}  & 78.3 & 67.9 \\
        \bottomrule
    \end{tabular}
    \caption{Novelty, Recall and Precision on ORCAS. \pgr obtains substantial gains in novelty with minimal accuracy drop.}
    \label{tab-recall-orcas}
\end{table*}


\subsection{Query-keyword recommendation task}
\paragraph{Novelty. } For the query-keyword dataset, the goal is to increase the novelty of policy's recommendations. For training \pgr, we use the reward from Equation~\ref{eq:rewardfn}. The results are shown in Table~\ref{tab-novelty}. 
Compared to supervised finetuning on the same training set, \pgr leads to 2X-5X gains in novel keywords in top-50. In top-200, \pgr (Conservative) and \pgr (Aggressive)  obtain 108 and 152 novel keywords respectively compared to 79 from the supervised finetuned model. At the same time, recall of \pgr (Conservative) is almost the same as the base model. In fact, recall@200 is slighly higher than the base model. To check the quality of the \pgr models, we also evaluate precision@3 as evaluated by GPT-4. While we would expect the precision to decrease due to the novelty loss, we find that \pgr (Conservative) has a substantially higher precision than the base model. This may be possible because the novelty loss can encourage the model to move away from local minima and sometimes find a more optimal solution. In comparison, \pgr (Aggressive) suffers a significant drop in both recall and precision, indicating that novelty optimization has led to a decrease in the model's accuracy. The offline results indicate that \pgr (Conservative) is a good balance between novelty and accuracy.    

Table~\ref{tab:qual} shows qualitative results for a sample of the queries where \pgr led to novel keywords in the top-5 predictions. The base model tends to match keywords based on lexical overlap wheareas \pgr is able to find rephrases with the same meaning. 

\begin{table*}[ht]
\centering
    \begin{tabular}{l|ccc|c}
        \toprule
        \textbf{Model} & \textbf{Novelty@10} & \textbf{Novelty@20} & \textbf{Novelty@50}  & \textbf{Recall@50} \\
        \midrule 
         Base Pretrained Model & -- &  -- & --  & {21.9} \\
        Supervised Finetuning  & 1.9 & 5.8 & 24.4    & \textbf{22.6} \\
        \pgr (Conservative) & \textbf{3.6} & \textbf{9.5}  & \textbf{32.4} &  22.4\\ 
        \bottomrule
    \end{tabular}
    \caption{Novelty and Recall on the AmazonReviews dataset. Compared to supervised finetuning, \pgr obtains substantial gains in novelty with almost the same recall.}
    \label{tab-recall-amazon}
\end{table*}

\paragraph{A/B test. }Finally, we evaluate \pgr (Conservative) on real user traffic using an A/B experiment over a 10 day period. As mentioned before, the retrieval system is engineered such that keywords from a new algorithm (e.g., \pgr) are selected for downstream ranking layer (which in turn, may lead to user impressions) only if they are novel compared to the existing algorithm. By including \pgr in the retrieval pipeline, we observe a $1$\% increase in query-ad matching density, the average number of relevant ads selected per query (as determined by the downstream ranker). We also observe a  $0.14$\% increase in coverage, the fraction of queries for which relevant ads are shown to users. This indicates that the novelty optimization can help to match ads to the queries that other algos are not able to match relevant ads with. Finally, we also observe a 0.26\% increase in click yield, the number
of ad clicks per search query, indicating the new keywords recommended by the novelty optimization are well received by the real-world  users. While the absolute number may look small,  
an increase of $0.26\%$  can lead to a substantial impact when scaled to millions of users.

\subsection{ORCAS: Query-webpage matching}
To simulate the production setting, we train a supervised model on over 16M \textless query, webpage title\textgreater pairs as the \textit{base} model. Our goal is to produce top-k webpages that are novel wrt. top-50 webpages predicted by the base model, using a training set of 1M pairs. We first use the InfoNCE loss to finetune the base model over the 1M training dataset (Supervised Finetuning). \pgr model is finetuned using our proposed method using the base model as the initialization. 

Table~\ref{tab-recall-orcas} shows the results. We report results for two supervised finetuning models (epochs 3 and 10). In both models, the recall and precision decreases compared to base model, indicating that the 1M train set may lead to overfitting. \pgr (Conservative) also leads to a drop in recall but the corresponding novelty is significantly higher than that of a supervised model with similar recall. At comparable recall, \pgr (Conservative) obtains Novelty@50 of 2.64 compared to 0.65 for the supervised model (3 epochs). Overall, \pgr (Aggressive) obtains the highest novelty---on average, there are $4.03$ webpages in  top-50 predictions of \pgr (Conservative) that did not exist in top-50 predictions of the base model. The novelty is significantly higher than a supervised model (2.82, 10 epochs) with similar recall. Recall and precision of \pgr models decrease slightly (up to 3\%) compared to base model, which can be a reasonable tradeoff given  the gains in novelty.

\subsection{Amazon: User-Product recommendation}
Table~\ref{tab-recall-amazon} shows the novelty and recall metrics for the AmazonReviews dataset. Both Supervised Finetuning and  \pgr models are initialized with the base pretrained RecFormer model and trained on the finetuning set. 
Since these models are trained on an additional recent product from the user's history, both models obtain slighly higher recall than the base model. However, at the same recall, novelty of the \pgr (Conservative) model is significantly higher than the supervised model. On average, \pgr recommends 3.6 products in its top-10 list that are novel wrt. the top-50 recommendations from the base model, compared to 1.9 novel products for the supervised model.   As in the previous datasets, \pgr leads to improved novelty   while incurring a minimal loss in recall compared to the supervised model.

\section{Conclusion}
We presented a technique to optimize a non-differentiable, task-specific loss in information retrieval applications. We justified the binary-action formulation of the problem through theoretical and empirical results. On empirical recommendation datasets, the proposed technique leads to substantial gains in novelty of top-k items.

While we used the simple REINFORCE-based RL algorithm, future work can consider Actor-critic or proximal policy optimization  algorithms for optimizing novelty and how  they can be extended to large action spaces. Exploring robust reward functions in the presence of noisy LLM feedback is also an important future direction.


\bibliographystyle{ACM-Reference-Format}
\balance
\bibliography{novelty-rl}

\clearpage
\appendix
\section{Appendix: Proof of Propositions}

\subsection{Proposition 1}
To prove Proposition 1, we restate a result from ~\cite{mei2020global} for a MDP. Here $\gamma$ is the discount factor,  $t$ refers to steps of the optimization, and $d^\pi_\mu(\gamma) =\mathbf{E}_{s_0 \sim \mu}(1 - \gamma) \sum_{t=0}^\infty \gamma^t P(s_t = s|s_0, \pi)$ is the discounted state distribution.

\begin{theorem}[\cite{mei2020global}]
Let Assumption 1 hold and let $\{\theta_t\}_{t \geq 1}$ be
generated using  $\theta_{t+1} \leftarrow \theta_t + \eta \frac{\partial V^{\pi_{\theta_t}}(\mu)}{\partial \theta_t}$ where $V^{\pi_{\theta_t}}(\mu) = \mathbf{E}_{s\sim \mu}  V^{\pi_{\theta_t}}(s)$. Let $\eta = (1 - \gamma)3/8$, $c$ be the
positive constant  $c := \inf_{s\in S,t\geq 1} \pi_{\theta_t}(a^*(s)|s) > 0$. Then, for all $t \geq 1$,
\begin{equation}
   \mathbf{E}_{s \sim \mu} [ V^*(s) - V^{\pi_{\theta_t}}(s)] \leq \frac{16S}{c^2 (1-\gamma)^6 t} \norm{\frac{d^{\pi*}_\mu (\gamma)}{\mu}}_{\infty}^2 \norm{\frac{1}{\mu}}_{\infty}
\end{equation}
\end{theorem}

\firstprop*
\begin{proof}
Note that our RL setup is one-step, hence we can assume $\gamma = 0$. Then $V^\pi(s)$ from Theorem 1 simplifies to, 
$$ V^\pi (s) = \mathbf{E}_{a_t \sim \pi(.|s_t)} \sum_{t=0}^\infty \gamma^t r(s_t, a_t) = \mathbf{E}_{a_t \sim \pi(.|s_t)} r(s_0, a_0) = \mathbf{\pi}(s)^T\mathbf{r}(s)$$ 
where $\pi$ and $r$ in the last equation refer to the vectors of action probabilities and their rewards given a state. Further, since $\gamma=0$, $d^{\pi*}_\mu (\gamma) = \mu$. Hence, we can write Theorem 1 as, 

\begin{align}
 \mathbf{E}_{s \sim \mu} [ V^*(s) - V^{\pi_{\theta_t}}(s)] & 
 = \mathbf{E}_{s \sim \mu} [\mathbf{\pi^*}(s)^T\mathbf{r}(s) - \mathbf{\pi_{\theta_t}}(s)^T\mathbf{r}(s)] \\
 & = \mathbf{E}_{s \sim \mu} [(\mathbf{\pi^*}(s) - \mathbf{\pi_{\theta_t}}(s))^T\mathbf{r}(s)] \\
 & = \mathbf{E}_{x \sim \mathcal{D}} [(\mathbf{\pi^*}(x) - \mathbf{\pi_{\theta_t}}(x))^T\mathbf{r}(x)] \\
 &\leq \frac{16S}{c^2 (1-\gamma)^6 t} \norm{\frac{d^{\pi*}_\mu (\gamma)}{\mu}}_{\infty}^2 \norm{\frac{1}{\mu}}_{\infty} \\
 &\leq \frac{16S}{c^2 t}  \norm{\frac{1}{\mu}}_{\infty}
 \end{align}

where the third equality is because the initial state distribution is the distribution of queries in the training data. 

Now, using Assumption 2, the minimum initial probability for the optimal action $a^*$ is $\frac{1}{\rho A}$ for all states. Assuming that the gradient updates do not decrease the probability of the optimal action, $c = \inf_{s\in S,t\geq 1} \pi_{\theta_t}(a^*(s)|s) = \frac{1}{\rho A}$.
Substituting c in the above equation, we obtain the result. 
$$ \mathbf{E}_{x \sim \mathcal{D}} [(\mathbf{\pi^*}(x) - \mathbf{\pi_{\theta_t}}(x))^T\mathbf{r}(x)] \leq \frac{16SA^2 \rho^2}{ t}  \norm{\frac{1}{\mu}}_{\infty} $$
\end{proof}

\subsection{Proposition 2}
\secondprop*
\begin{proof}
Using the equation from Proposition 1 and substituting $S = SA$ and $A=2$ leads us to the result. 
\end{proof}

\end{document}